\documentclass{IEEEtran}

\usepackage{cite}
\usepackage{amssymb,amsfonts}
\usepackage{algorithmic}
\usepackage{graphicx}
\usepackage{textcomp}
\usepackage{wrapfig}
\usepackage{graphicx}
\usepackage{graphics} % for pdf, bitmapped graphics files
\usepackage{epsfig} % for postscript graphics files
\usepackage{caption}
\usepackage{subcaption}
\usepackage{enumitem}
\usepackage{amsthm}
\usepackage[dvipsnames]{color}
\usepackage[fleqn]{amsmath}
\usepackage{tikz}
\usetikzlibrary{shapes,arrows}
\usetikzlibrary{positioning}
\usetikzlibrary{arrows.meta}
\usetikzlibrary{calc}
\usetikzlibrary{fit}
\usepackage{pgfplots}

\pgfplotsset{compat=newest} 
\usepackage{pgfplotstable}
\usepackage{lipsum}  

\usetikzlibrary{external}
\def\BibTeX{{\rm B\kern-.05em{\sc i\kern-.025em b}\kern-.08em
		T\kern-.1667em\lower.7ex\hbox{E}\kern-.125emX}}

\newcommand{\R}{\mathbb{R}}
\newcommand{\N}{\mathbb{N}}
\newcommand{\C}{\mathcal{C}}

\newcommand{\V}{\mathcal{V}}
\newcommand{\W}{\mathcal{W}}
\newcommand{\B}{\mathcal{B}}
\newcommand{\Q}{\mathcal{Q}}

\newtheorem{theorem}{Theorem}
\newtheorem{assumption}{Assumption}
\newtheorem{proposition}{Proposition}
\newtheorem{lemma}{Lemma}
\newtheorem{definition}{Definition}
\newtheorem{remark}{Remark}

\newcommand{\norm}[1]{\left\lVert#1\right\rVert}
\newcommand{\abs}[1]{\left\lvert#1\right\rvert}

\begin{document}
	\tikzstyle{block} = [draw, rectangle, 
	minimum height=3em]
	\tikzstyle{sum} = [draw, circle, scale = 0.5,node distance = 1cm]
	\tikzstyle{product} = [draw, circle, cross, minimum width=1 cm, scale = 0.5]
	\tikzstyle{input} = [coordinate]
	\tikzstyle{output} = [coordinate]
	\tikzstyle{pinstyle} = [pin edge={to-,thin,black}]
	\tikzstyle{gain} gain  = [draw, thick, isosceles triangle, minimum height = 2em, isosceles triangle apex angle=60, scale = 0.8]
	\bibliographystyle{IEEEtranS}

	\title{Newton Nonholonomic Source Seeking\\ for Distance-Dependent Maps}
	\author{Velimir Todorovski,  Miroslav Krsti\'{c}, \IEEEmembership{Fellow, IEEE}
		\thanks{Velimir Todorovski is with the Faculty of Electrical and Computer Engineering, Technical University of Munich, 80333 Munich, Germany (e-mail: velimir.todorovski@tum.de).}
		\thanks{Miroslav Krstic is with the Department of Mechanical and Aerospace Engineering University of California at San Diego, La Jolla, CA 92093 (e-mail: krstic@ucsd.edu).}}
	
	\maketitle
	
	\begin{abstract}
		The topics of source seeking and Newton-based extremum seeking have flourished, independently, but never combined. We present the first Newton-based source seeking algorithm. The algorithm employs forward velocity tuning, as in the very first source seeker for the unicycle, and incorporates an additional Riccati filter for inverting the Hessian inverse and feeding it into the demodulation signal. Using second-order Lie bracket averaging, we prove convergence to the source at a rate that is independent of the unknown Hessian of the map. The result is semiglobal and practical, for a  map that is quadratic in the distance from the source. The paper presents a theory and simulations, which show advantage of the Newton-based over the gradient-based source seeking. 
	\end{abstract}
	
	%\begin{IEEEkeywords}
	%Enter key words or phrases in alphabetical 
	%order, separated by commas. For a list of suggested keywords, send a blank 
	%e-mail to keywords@ieee.org or visit \underline
	%{http://www.ieee.org/organizations/pubs/ani\_prod/keywrd98.txt}
	%\end{IEEEkeywords}
	
	\section{Introduction}
	\label{sec:introduction}
	
	%	The initial ES schemes \cite{ES_book} were significantly influenced by the unknown Hessian, which limited their convergence speed. However, substantial progress has been made in overcoming this limitation through the development of Newton-based ES methods . This approach decouples the convergence rate from the Hessian, enabling the assignment of arbitrary convergence rates by appropriately selecting the gain. As a result, there have been numerous applications and theoretical extensions that have greatly improved the convergence speed and overall performance of ES algorithms. 
	%	
	
	We present the first Newton-based source seeking algorithm, making the convergence rate invariant to the unknown Hessian of the map. This result has been sought for well over a decade (since \cite{6160495}) but was unattainable prior to the theoretical breakthrough in \cite{labar2019newton}.
	
	\textit{Background:}
	The initial extremum seeking (ES) schemes \cite{ES_book} were gradient based. As such, while convergent, their convergence rate was unknown to the user and unassignable, since it depended on the unknown Hessian of the map. This limitation was overcome with the Newton-based ES algorithm \cite{nesic2010newton,newton_ES}. These algorithms remove the Hessian from the linearization of the average system, enabling the assignment of arbitrary convergence rates by appropriately selecting the gain. Due to this advantage, various extensions and applications have been proposed such as stochastic Newton-based ES \cite{stochastic_newton}, Newton-based ES with different time delays \cite{damir_newton}, \cite{damir_newton2}, Newton-based ES for partial differential equations \cite{PDE_newton}, fixed-time Newton-based ES schemes \cite{fixed_newton} and optimization of power for photovoltaic microconverters \cite{power_newton}.
	
	Several years prior to the efforts in Newton-based ES, numerous source seeking algorithms began to emerge \cite{first_ss,zhang2007source,cochran2009nonholonomic,ES_with_bounded_update_rates,liu2010stochastic,lin2017stochastic,suttner2019extremum,suttner2020acceleration,suttner_torque,ss_planar,suttner2023nonlocal,9782675,abdelgalil2022sea}, where the primary focus was designing algorithms that ensure convergence to the source using a (kinematically) underactuated, non-holonomic unicycle, with limited attention given to the convergence speed.
	This issue of convergence speed was somewhat addressed in our work \cite{todorovski2023practical}, with an algorithm that ensures convergence to the source in {\em prescribed time}, despite the influence the unknown Hessian, by introducing time-varying gains. However, practical implementation of this approach remains a challenge, due to the need for the vehicle speed to grow. Hence, all the source seeking algorithms in existence have convergence rates that  depend on the unknown Hessian and are not user-assignable.
	
	\textit{Contribution:} In this paper, we begin to fill the need for Newton source seekers. Building on the initial source seeking scheme in \cite{zhang2007source}, our approach incorporates a Newton-based modification that tunes the forward velocity while keeping the angular velocity constant.
	Borrowing concepts from the Newton-based ES literature, we estimate the inverse of the Hessian of the map, on average, in a model-free manner via a Riccati filter. Effectively, this allows for  decoupling the convergence rate from the Hessian, enabling us to assign arbitrary convergence rates through the gain parameter. 
	
	We establish convergence to the source with the relatively recently pioneered second-order Lie bracket averaging \cite{labar2019newton}, guaranteeing practical asymptotic stability semiglobally. Finally, we validate our theory  through a simulation example.

		\textit{Organization:}	In Section \ref{sec:preliminaries}, we provide an overview of a second-order Lie bracket method, which is used for stability analysis of input-affine systems subjected to oscillatory inputs. We then proceed to Section \ref{sec:problem_statement}, where we define the general problem of unicycle source seeking and review an approach to tackle it in Section \ref{sec:gradient_source_seeking}.
		Our main contribution is presented in Section \ref{sec:newton_source_seeking}, where we propose a novel algorithm for Newton-based source seeking. We present the simulations results in Section \ref{sec:simulation}. 
		
		\textit{Notation:} The Euclidian norm is denoted as $\abs{\cdot}$ and class $\mathcal{K}$, $\mathcal{K}_{\infty}$ and $\mathcal{KL}$ functions are defined as in \cite{khalil2002nonlinear}. $\R_{>0}$ and $\mathbb{Q}_{>0}$ denote the sets of strictly positive real and rational numbers, respectively. \textbf{LCM}$(k_1,k_2,\ldots,k_n)$ stands for the Least Common Multiple of $(k_1,k_2,\ldots,k_n)$. The Jacobian of a differentiabe vector field $f: \R^n \rightarrow \R^n$ is denoted by $\frac{\partial f}{\partial x}$. The Lie bracket of two differentiable vector field $f: \R^n \rightarrow \R^n$ and $b: \R^n \rightarrow \R^n$, denoted by $[f,g]$ is defined by  $\frac{\partial g}{\partial x}f(x) - \frac{\partial f}{\partial x}g(x)$. A neighborhood of a point $x \in \R^n$ is a connected open set containing a ball centered in the point $x$.
		
		%Due to the space limitation, some computations are omitted, however they are provided in an arXiv supplementary document to this paper. A

		\section{Preliminaries: Second-order Lie bracket approximations }
		\label{sec:preliminaries}
		
		For the purpose of developing this paper's results, we start by calling to mind the method in \cite{labar2019newton} for stability analysis of control-affine systems with oscillatory inputs. 
		
		Consider a system of the form 
		\begin{equation}
			\dot{x}(t) =  f_0(x(t))  + \sum\limits_{i=1}^{l} f_i(x(t))\omega^{p_i} u_i(k_i \omega t) 	\label{eq:general_nominal_system}
		\end{equation}
		with the state $x \in \mathbb{R}^n$, the control inputs $u_i(t) \in \R$ and the parameters $\omega \in \R_{>0}$, $k_i \in \mathbb{Q}_{>0}$, $l \in \N$, $p_i \in (0,1)$. Furthermore, suppose that the following assumptions on $f_i$ and $u_i$ hold. 
		\begin{assumption}
			\label{ass:second_order_ass1}
			For all $i = 0,\ldots,l$ and $j = 1,\ldots,l$:
			\begin{enumerate}
				\item  $f_i \in \C^3: \mathbb{R}^n \times \mathbb{R}^n $.
				\item $u_j(t): \R \rightarrow \R$ is a measurable bounded function where it is assumed that $\sup_{t\in\R} \abs{u_j(t)}\le 1$.
				\item $u_j(t)$ is $2 \pi$-periodic, i.e., $u_j(t) = u_j(t+2\pi)$, $\forall t \in \R$ and has a zero mean on one period, i.e.,$\int_{0}^{2 \pi} u_j(t) \text{d}t = 0$.
			\end{enumerate}
		\end{assumption}
		
		\begin{assumption}
			\label{ass:second_order_ass2}
			Let $T_{i j}$ be a common period to $u_i(k_i \omega t)$ and $u_j(k_j \omega t)$ and let $T_{i j m}$ be a common period to $u_i(k_i \omega t)$, $u_j(k_j \omega t)$ and $u_m(k_m \omega t)$ with $i,j,m \in \{1,2,\ldots,l\}$. If $p_i + p_j > 1$, then $[f_i,f_j](x) = 0$, $\forall x \in \R^n$ or $\int_{0}^{T_{ij}} \int_{0}^{s} u_j(k_j\omega s) u_i(k_i \omega p) \text{d} p \text{d} s = 0$. Furthermore, if $p_i + p_j + p_m > 2$, then $[[f_i,f_j],f_m](x) = 0$, $\forall x \in \R^n$ or $\int_{0}^{T_{ijm}} u_m(k_m \omega \tau) \int_{0}^{T_{ij}} \int_{0}^{s} (u_j(k_j\omega s) u_i(k_i \omega p) - u_i(k_i\omega s) u_j(k_j \omega p) ) \text{d} p \text{d} s \text{d} \tau = 0$. Finally, if $p_i + p_j + p_m + p_q \ge 3$, then $\frac{\partial}{\partial x} \left(\frac{\partial}{\partial x} \left(\frac{\partial f_i(x)}{\partial x} f_j(x)\right)f_m(x)\right)f_q(x) = 0$, $\forall x \in \R^n$.
		\end{assumption}
		According to \cite[Lemma 2]{labar2019newton} the trajectories of \eqref{eq:general_nominal_system} are approximated by the so-called second-order Lie bracket system associated with \eqref{eq:general_nominal_system} that has the following form
		\begin{equation}
			\label{eq:second_order_LBS_general}
			\begin{aligned}[b]
				\dot{\bar{x}} = f_0(\bar{x}) &+ \lim\limits_{\omega \rightarrow \infty} \sum\limits_{\substack{1 \le i < l \\ i<j\le l}} [f_i,f_j] \gamma_{ij}(\omega)\\ &+ \lim\limits_{\omega \rightarrow \infty}\sum\limits_{\substack{1 \le i < l \\ i<j\le l}} \sum\limits_{m = 1}^l [[f_i,f_j],f_m] \gamma_{ijm}(\omega) 
			\end{aligned}
		\end{equation}
		where 
		\begin{eqnarray}
			\gamma_{ij}(\omega) &=& \frac{\omega^{p_i + p_j}}{T} \int\limits_{0}^{T} \int\limits_{0}^{s} u_j (k_j \omega s) u_i(k_i \omega p) \text{d} p \text{d} s
			\\
			\gamma_{ijm}(\omega) &= & \frac{\omega^{p_i + p_j + p_m}}{T} \int\limits_{0}^{\tau} \int\limits_{0}^{s} u_m (k_m \omega \tau) 
			\nonumber\\ && 
			\times\int\limits_{0}^{T} \Big[ u_j (k_j \omega s) u_i(k_i \omega p)  
			\nonumber\\ && 
			- u_i (k_i \omega s) u_j(k_j \omega p) \Big] \text{d} p \text{d} s \text{d} \tau 
		\end{eqnarray}
		with $T = \frac{2 \pi}{\omega}\text{\textbf{LCM}}(k_1^{-1},..,k_l^{-1})$.
		
		As \eqref{eq:second_order_LBS_general} is an approximation of \eqref{eq:general_nominal_system}, we can expect that both systems share some stability properties. Roughly stated, asymptotic stability of \eqref{eq:second_order_LBS_general} implies \textit{practical} asymptotic stability of \eqref{eq:general_nominal_system}. In this particular case, the notion of practical stability is defined as follows.
		\begin{definition}
			\label{def:sGPUAS} The origin of the $n$-dimensional system $\dot{x} = f(t,x,\epsilon)$ with the vector of parameters $\epsilon = [\epsilon_1,\epsilon_2,\ldots,\epsilon_{n_{\epsilon}}]^\top$, is semi-Globally Practically Uniformly Asymptotically Stable (\textbf{sGPUAS}) if for every bounded neighborhood of the origin $\B \subset \R^n$ and $\V \subset \R^n$, there exists bounded neighborhoods of the origin $\Q \subset \B$ and $\W: \V \subset \W \subset\R^n$ and an $\epsilon_1^* \in \R_{>0} $ such that $\forall \epsilon_1 \in (0, \epsilon_1^*)$, $\exists \epsilon_2^*$, such that $\forall \epsilon_2 \in (0, \epsilon_2^*) $,$\dots$,$\exists \epsilon_{n_{\epsilon}}^*$, such that $\forall \epsilon_{n_{\epsilon}}^* \in (0, \epsilon_{n_{\epsilon}}^*) $, $\exists t_1 \in \R$ such that $\forall t_0 \in \R$ the following holds:
			\begin{enumerate}
				\item[1.] (Boundedness) $x_0 \in \V \Rightarrow x(t) \in \W, \forall t\ge t_0$;
				\item[2.] (Stability) $x_0 \in \Q \Rightarrow x(t) \in \B, \forall t \ge t_0$;
				\item[3.] (Pract. Convergence) $x_0 \in \V \Rightarrow x(t) \in \B, \forall t \ge t_1 + t_0 $.
			\end{enumerate}
		\end{definition}
		Thus, \cite[Lemma 3]{labar2019newton} summarizes the stability properties of \eqref{eq:general_nominal_system} and \eqref{eq:second_order_LBS_general}, which for the sake of completeness we state here 
		\begin{lemma} 
			\label{lemma:sGPUAS}
			Under Assumptions \ref{ass:second_order_ass1} and \ref{ass:second_order_ass2}, if the origin of \eqref{eq:second_order_LBS_general} is globally uniformly asymptotically stable (\textbf{GUAS})\footnote{GUAS definition can be found in \cite[Lemma 4.5]{khalil2002nonlinear}.} with a vector of parameters $\epsilon = [\epsilon_1,\epsilon_2,\ldots,\epsilon_{n_{\epsilon}}]^\top$, then the origin of \eqref{eq:general_nominal_system} is \textbf{sGPUAS} as in Def. \ref{def:sGPUAS} with the vector parameters $[\epsilon^\top,\omega^{-1}]^\top$. 
		\end{lemma}

		\section{Unicycle Source Seeking Problem}
		\label{sec:problem_statement}
		%		\begin{figure}
			%			\centering
			%			\scalebox{0.8}{
				%				\begin{tikzpicture}
					%					\begin{axis}[
						%						scale = 0.8,
						%						axis lines = left,
						%						xmin = 0,
						%						xmax = 1.5,
						%						ymin = 0,
						%						ymax = 1.5,
						%						axis line style={very thick},
						%						xlabel=$x_1$,
						%						ylabel=$x_2$,
						%						every axis x label/.style={at={(axis description cs:1,-0.03)},anchor=north},
						%						every axis y label/.style={at={(current axis.above origin)},anchor=north east},
						%						ticks = {none}]
						%						
						%						\begin{scope}[xshift = 1cm, yshift = -1cm]
							%							\draw[thick,rotate around={40:(0.75,0.75)}, fill=gray!10, rounded corners] (0.2, 1) rectangle (0.7,1.45);
							%						\end{scope}
						%						
						%						\begin{scope}[xshift = 0.5cm, yshift = 0.55cm]
							%							\draw[thick,rotate around={40:(0.75,0.75)}, fill=black, rounded corners] (0.2, 0.4) rectangle (0.35,0.5);
							%						\end{scope}
						%						
						%						\begin{scope}[xshift = -0.6cm, yshift = 1.85cm]
							%							\draw[thick,rotate around={40:(0.75,0.75)}, fill=black, rounded corners] (0.2, 0.4) rectangle (0.35,0.5);
							%						\end{scope}
						%						
						%						\draw [thick, dashed] (axis cs:0.8,0) node[xshift = 0.1cm,yshift = 0.7cm] {$\theta$} arc [radius=0.8,start angle=0,end angle=45];
						%						
						%						\draw[dashed,thick] (0,0) -- (0.55,0.55);
						%						
						%						\draw[->,ultra thick] (0.57,0.57) -- node[xshift = 0.1cm, yshift = 0.6cm] {$u_1$} (1,1);
						%					\end{axis}
					%				\end{tikzpicture}
				%			}
			%			\caption{ Model of a nonholonomic unicycle with forward velocity $u_1$ and orientation $\theta$.}
			%			\label{fig:unicycle}
			%		\end{figure}
		
		Let the nonholonomic unicycle be modeled by % in Fig. \ref{fig:unicycle}
		\begin{subequations}
			\label{eq:unicycle_model}
			\begin{align} 
				\dot{x} &= u_1\begin{bmatrix}
					\cos(\theta) \\
					\sin(\theta)
				\end{bmatrix} 
				\label{eq:unicycle_with_drift-a}
				\\
				\dot{\theta} &= u_2 \\
				y  &= F(x)
			\end{align}
		\end{subequations}
		where $x = [x_1,x_2]^{\top} \in \R^2$ is the position of the vehicle's center with  $x(t_0) = x_0$, $\theta \in \R$ is the orientation with $\theta(t_0) = \theta_0$ and $u_1(t), u_2(t) \in \R$ are the forward and angular velocity inputs, respectively. Let the measurement $y \in \R$ of the field $F(x)$ at position $x$ of the vehicle be available. 
		
		In relation to developing a Newton source seeker, in this note we assume that the signal field $F(x)$ decays with the distance from the source. Electromagnetic and acoustic signals  (in the absence of reflections) have this property. In view of our goal for a Newton strategy, and of freeing only the local convergence of the unknown Hessian, we make the following simplifying assumption. 
		
		\begin{assumption}[Distance-dependent field]
			\label{ass:shape_of_field}
			The function $F: \R^2 \rightarrow \R$ has the following form
			\begin{equation}
				F(x) = F^* - \frac{1}{2}H \abs{x-x^*}^2
				\label{eq:shapeofField}
			\end{equation}
			where the maximum at position $x^*$ is also referred to as 'source' and $H > \R_{>0}$ is the unknown Hessian parameter.  
		\end{assumption}
		The objective is to design a control law $u_1(t)$, $u_2(t)$ for  \eqref{eq:unicycle_model}, so that the position of the vehicle $x$ converges to the source $x^*$ of the field $F(x)$ without using position information $x$.
		
		\section{Gradient-based Source Seeking}
		\label{sec:gradient_source_seeking}
		Zhang et al. \cite{zhang2007source} proposed the first unicycle source seeking solution, in which the forward velocity is tuned, while the angular velocity is kept constant. As a result of this, the closed-loop dynamics, in an average sense, approximates the gradient of the field $F$ and the unicycle converges to the source. A major limitation of this algorithm is that its convergence rate is dependent on the unknown Hessian $\nabla^2 F(x)$. % However, our main result addresses this limitation and eliminates the dependence of the Hessian on the convergence rate when Assumption \ref{ass:shape_of_field} is satisfied. %
		Since this is neither obvious nor elucidated by the conventional (first-order) averaging analysis in \cite{zhang2007source}, in this section we review the properties of this algorithm from the perspective of second-order Lie bracket approximations. Consider the source seeking algorithm \cite{zhang2007source},
		\begin{subequations}
			\label{eq:gradient_seeker_control_law}
			\begin{align}
				u_1 &= c (F(x) - \nu) \sin(\omega t)  + \tilde{\alpha} \cos(\omega t)\\
				u_2 &= \omega_0 \\
				\dot{\nu} &= h(F(x) - \nu)
			\end{align}
		\end{subequations}
		where $c,\tilde{\alpha}, h, \omega, \omega_0 \in \R_{>0}$. Substituting \eqref{eq:gradient_seeker_control_law} in \eqref{eq:unicycle_model}, and using the fact that $\dot{\theta} = \omega_0 \Rightarrow \theta = \omega_0 t$, results in the following closed-loop system for source seeking
		{\setlength{\mathindent}{0pt}
			\begin{subequations}
				\label{eq:gradient_seeker}
				\begin{align}
					\dot{x} &= \Big[c (F(x) - \nu) \sin(\omega t)  + \tilde{\alpha}  \cos(\omega t)\Big] \begin{bmatrix}
						\cos(\omega_0 t) \\ \sin(\omega_0 t) 
					\end{bmatrix} \\
					\dot{\nu} &= h(F(x) - \nu).
				\end{align}
			\end{subequations} 
		}
		Two different periodic signals appear in here, one of frequency $\omega$ and the other of frequency $\omega_0$, and therefore the convergence of $x$ to the source $x^*$ is not readily apparent. 
		%from the source seeking scheme \eqref{eq:gradient_seeker}. 
		%To gain insight into the convergence behaviour, we need to examine the average dynamics of the system. However, s
		Furthermore, since the vector fields in \eqref{eq:gradient_seeker} are nonautonomous, the method from Section \ref{sec:preliminaries} is not directly applicable. For this reason, we perform the following linear time-varying  transformation of the position:
		\begin{align}
			z = Y^{\top}(t) (x-x^*) \label{eq:z_transformation}
		\end{align}
		where 
		\begin{align}
			Y(t) = \begin{bmatrix}
				\sin(\omega_0 t) &\cos(\omega_0 t) \\
				-\cos(\omega_0 t) &\sin(\omega_0 t) 
			\end{bmatrix}.
		\end{align}  
		We will need a differential equation for this new position and note that
		\begin{align}
			\dot{z} = \dot{Y}^{\top}(t)(x-x^*) + Y^{\top}(t)\dot{x}.
		\end{align}
		In view of the identities
		\begin{align}
			\label{eq-inverse-of-x-to-z}
			x-x^* &= Y(t) z
			\\
			Y^{\top}(t) Y(t) &= Y(t) Y^{\top}(t) = I 
			\\
			\dot{Y}^{\top}(t) &= 
			% Y^{\top}(t) \begin{bmatrix}
				% 0 & \omega_0\\
				% -\omega_0 &0
				% \end{bmatrix} = 
			Y^{\top}(t) J_0 
			\\
			J_0 &=  \begin{bmatrix}
				0 & \omega_0\\
				-\omega_0 &0
			\end{bmatrix} 
			\label{eq:J_0_def} 
			\\
			Y^{\top}(t) J_0 Y(t) &= J_0 
			\\
			F(z) &= F^* - \frac{1}{2} H \abs{z}^2
		\end{align}
		where $I$ is the identity matrix, we obtain the following  system from \eqref{eq:gradient_seeker}: 
		{\setlength{\mathindent}{0pt}
			\begin{subequations}
				\label{eq:gradient_seeker_in_z}
				\begin{align}
					\dot{z} &= J_0 z + \Big[c (F(z) - \nu) \sin(\omega t)  + \tilde{\alpha}  \cos(\omega t)\Big] \begin{bmatrix}
						0 \\
						1
					\end{bmatrix} \\
					\dot{\nu} &= h(F(z) - \nu).
				\end{align}
			\end{subequations}
		}The averaged system of \eqref{eq:gradient_seeker_in_z} is stated in the following proposition. 
		\begin{proposition}
			\label{prop:average_gradient_seeker}
			Consider \eqref{eq:gradient_seeker_in_z} and let $c = \omega^{1-p}$ and $\tilde{\alpha} = \alpha \omega^{p}$ where $p \in (0.5,1)$. Then, under Assumption \ref{ass:shape_of_field}, its second-order Lie bracket averaged system is given as 
			\begin{subequations}
				\label{eq:LBS_gradient_seeker}
				\begin{align}
					\dot{\bar{z}} &=  (J_0  + \Lambda) \bar{z}   \\
					\dot{\bar{\nu}} &= h(F(\bar{z}) - \bar{\nu}).
				\end{align}
			\end{subequations}
			where 
			\begin{align}
				F(\bar{z}) &= F^* - \frac{1}{2} H \abs{z}^2 \label{eq:field_in_z} \\
				\Lambda &= \begin{bmatrix}
					0 &0 \\
					0 & -\frac{\alpha H}{2}
				\end{bmatrix} \label{eq:Lambda_def}
			\end{align} and $J_0$ is defined as in \eqref{eq:J_0_def} .
		\end{proposition}
		The proof of Propostion \ref{prop:average_gradient_seeker} is found in the Appendix. Note that, similar results to Lemma \ref{prop:average_gradient_seeker} based on different averaging techniques have been reported in \cite{zhang2007source} and \cite{LBA_first}.
		
		By linearizing \eqref{eq:LBS_gradient_seeker} around $\bar{z}^* = 0, \bar{\nu} = F^*$, we obtain 
		\begin{align}
			\begin{bmatrix}
				\dot{\bar{z}} \\
				\dot{\bar{\nu}}
			\end{bmatrix} = \begin{bmatrix}
				0 &\omega_0 &0 \\
				-\omega_0 &-\frac{\alpha H}{2} &0 \\
				0 &0 &-h
			\end{bmatrix} 	\begin{bmatrix}
				\bar{z} \\
				\bar{\nu}
			\end{bmatrix} = J_{\text{g}} \begin{bmatrix}
				\bar{z} \\
				\bar{\nu}
			\end{bmatrix} \label{eq:gradient_seeker_lin}
		\end{align} where the eigenvalues of \eqref{eq:gradient_seeker_lin} are, 
		\begin{align}
			\lambda_{1,2}(J_{\text{g}}) &= -\frac{\alpha H}{2} \pm \frac{1}{2} \sqrt{ \left(\frac{\alpha H}{2}\right)^2 - (2\omega_0)^2} \label{eq:ev_zbar_grad}\\
			\lambda_{3}(J_{\text{g}}) &= -h.
		\end{align}
		Inspecting \eqref{eq:ev_zbar_grad}, we can easily see that although $\bar{z} \rightarrow 0$, its convergence is governed by the unknown Hessian $H$. In the next section, we address this limitation and eliminate the dependence of the Hessian on the convergence rate.

		\section{Newton-based Source Seeking}
		\label{sec:newton_source_seeking}
		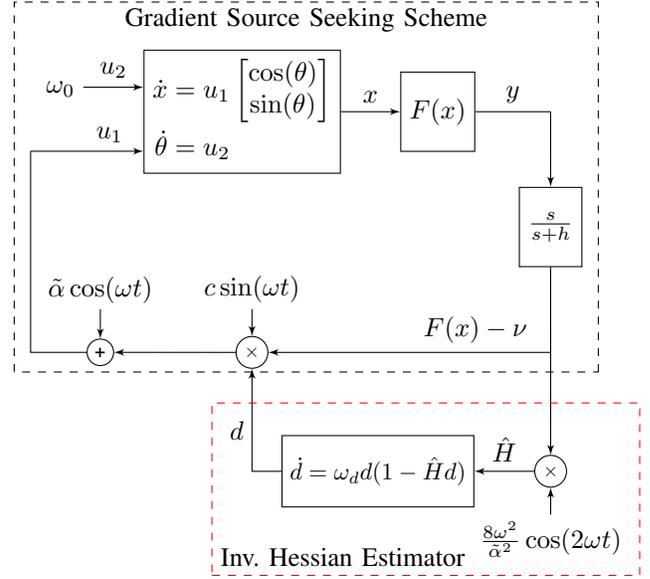
\begin{figure}[t]
			\centering
			
			\begin{tikzpicture}[auto, node distance=2cm,>=latex',scale = 0.6]
				% We start by placing the blocks
				
				\node [block ] (system) {$\begin{aligned}
						\dot{x} &=  u_1 \begin{bmatrix}
							\cos(\theta) \\
							\sin(\theta)
						\end{bmatrix} \\
						\dot{\theta} &= u_2
					\end{aligned}$};

				\node[input, above left= -0.5 cm and 0.8cm of system] (help)  {};
				
				\node[input, left of = help, node distance = 0cm] (help2)  {};
				
				\node[input, below left= -0.3cm and 1.5 cm of system] (in) {};
				\node[input, below left= -0.3cm and 0 cm of system, node distance = 0cm] (helpin) {};
				
				\node [block, right of = system, node distance = 2.6cm] (Fx) {$F(x)$};
				
				\draw [->] (system) -- node[name=u] {$x$} (Fx);
				
				% We draw an edge between the controller and system block to 
				% calculate the coordinate u. We need it to place the measurement block. 
				
				\node [output, right of = Fx, node distance = 1.5cm] (output) {};
				
				\node [sum, below left = 3.05cm  and 1.2cm of u ] (gain) {\Large $\mathbf{\times}$};
				
				\node[above of = gain, node distance = 0.7cm] (sin) {};
				
				\node [sum, below right= 3.4cm and 0.1cm of help, node distance = 2.5cm] (sum1) {\Large \textbf{+}};
				
				\node [above of = sum1, node distance = 0.7cm] (omega) {};

				\draw [->] (help) --  node[anchor = south] {$u_2$} (system.west |- help);
				\draw [] (help2) -- node[xshift = -0.3cm,yshift = -0.25cm] {$\omega_0$} (help);
				\draw [->] (in) -- node[anchor = south, pos = 0.7] {$u_1$} (helpin);
				\draw [] (Fx) -- node {$y$} (output);

				\node[block, below  = 1cm of output, node distance  = 0.5cm] (filter) {$\frac{s}{ s + h}$};

				\draw [->] (output) -- (filter);
				
				\draw [->] (filter) |- node[anchor = south, xshift = -1cm] {$F(x) - \nu$} (gain);
				\draw [->] (gain) -- (sum1);
				\draw [->] (omega) -- node[anchor = south, pos = 0.1] {$\tilde{\alpha} \cos(\omega t)$} (sum1);
				\draw [->] (sin) -- node[anchor = south, pos = 0.1] {$c \sin(\omega t)$} (gain);
				\draw[-] (sum1) -| (in);
				
				\draw[dashed] ($(current bounding box.south east)+(10pt,-2pt)$) rectangle node[anchor = north, yshift = 2.5cm] {Gradient Source Seeking Scheme} ($(current bounding box.north west)+(-10pt,30pt)$);
				
				\node [sum, below = 2.5cm of filter ] (product) {\Large $\mathbf{\times}$};
				\node[block, left = 0.8cm of product,scale = 0.9] (ricatti) {$\dot{d} = \omega_d d(1-\hat{H}d)$};
				\node [below of = product, node distance = 0.7cm] (costwo) {};

				\draw[->] (filter) -- (product);
				\draw[->] (ricatti) -| node[anchor = south,xshift = -0.2cm,yshift = 0.3cm] (d) {$d$} (gain);
				\draw[->] (product) -- node[anchor = south] {$\hat{H}$}(ricatti);
				\draw[->] (costwo) -- node[anchor = north, pos = 0.1] (cos2){$ \frac{8 \omega^2}{\tilde{\alpha}^2} \cos(2 \omega t)$} (product);

				\node [fit=(d) (cos2),draw,dashed, red, label={[anchor=south west]south west:Inv. Hessian Estimator}] {};
				
			\end{tikzpicture}
			\caption{Newton-based Source Seeking Scheme.} 	
			\label{fig:newton_seeker}
		\end{figure}
		In the previous section, we demonstrated that the $x$-trajectories of \eqref{eq:gradient_seeker} converge to a neighborhood around the source $x^*$. However, the rate of convergence is dependent on the unknown Hessian $H$. We provide next a control design with a subsequent analysis where this dependency is removed for signals that satisfy Assumption \ref{ass:shape_of_field}. For this, we employ the inverse Hessian estimation technique for Newton-based ES presented by Ghaffari et al. \cite{newton_ES} and establish the stability properties of the developed closed-loop system via the second-order Lie bracket averaging technique presented in Section \ref{sec:preliminaries}.  
		
		\subsection{Design}
		
		Our Newton modification of the gradient source seeking scheme \eqref{eq:gradient_seeker} is shown in Fig. \ref{fig:newton_seeker} and given, in equation form, by
		\begin{subequations}
			\label{eq:newton_seeker_control_law}
			\begin{align}
				u_1 &= c d (F(x) - \nu) \sin(\omega t)  + \tilde{\alpha} \cos(\omega t)\\
				u_2 &= \omega_0 \\
				\dot{d} &= \omega_d d\left(1 - d \frac{8 \omega^2}{\tilde{\alpha}^2}(F(x) - \nu) \cos(2\omega t) \right) \\
				\dot{\nu} &= h(F(x) - \nu)
			\end{align}
		\end{subequations}
		where $\omega_d \in \R_{>0}$ and $d \in \R$ with $d(0) = d_0$. The closed-loop system, obtained similarly to \eqref{eq:gradient_seeker}, reads as
		{\setlength{\mathindent}{0pt}
			\begin{subequations}
				\label{eq:newton_seeker}
				\begin{align}
					\dot{x} &= \Big[ c d(F(x) - \nu) \sin(\omega t)  + \tilde{\alpha}  \cos(\omega t)\Big]  \begin{bmatrix}
						\cos(\omega_0 t) \\ \sin(\omega_0 t) 
					\end{bmatrix} \label{eq:newton_seeker_a}\\
					\dot{d} &= \omega_d d\left(1 - d \frac{8 \omega^2}{\tilde{\alpha}^2}(F(x) - \nu) \cos(2\omega t) \right) \\
					\dot{\nu} &= h(F(x) - \nu)
				\end{align}
			\end{subequations}
		}
		The basic idea behind the Newton source seeking scheme \eqref{eq:newton_seeker} can be summarized as follows. The high-pass filtered signal $F(x) - \nu$ is multiplied with the perturbation $\cos(2\omega t)$ to obtain an average estimate of the Hessian. This signal is then passed through a Riccati filter $d$, which produces an estimate of the inverse Hessian that remains valid even when the Hessian is singular. Finally, the output of the Riccati filter, $d$, is fed back into the system via the forward velocity. The convergence of \eqref{eq:newton_seeker_a} to the source $x^*$ is independent of the Hessian which can be seen by deriving its average system, which we do next. 
		%inspecting the average of \eqref{eq:newton_seeker_in_z}. 
		Just as before, in order to use the technique from Section \ref{sec:preliminaries}, we transform  \eqref{eq:newton_seeker_a} with \eqref{eq:z_transformation} and obtain 
		{\setlength{\mathindent}{0pt}
			\begin{subequations}
				\label{eq:newton_seeker_in_z}
				\begin{align}
					\dot{z} &= J_0 z + \Big[ c d (F(z) - \nu) \sin(\omega t)  + \tilde{\alpha}  \cos(\omega t)\Big]  \begin{bmatrix}
						0 \\
						1
					\end{bmatrix} \\
					\dot{d} &= \omega_d d\left(1 - d \frac{8 \omega^2}{\tilde{\alpha}^2}(F(z) - \nu) \cos(2\omega t) \right) \label{eq:ricatti_filter_z}\\
					\dot{\nu} &= h(F(z) - \nu).
				\end{align}
			\end{subequations}
		}
		The averaged system is given by the next proposition.
		\begin{proposition}
			\label{prop:average_newton_seeker}
			Consider \eqref{eq:newton_seeker_in_z} and let $c = \omega^{1-p}$ and $\tilde{\alpha} = \alpha \omega^{p}$ where $p \in (0.5,1)$. Then, under Assumption \ref{ass:shape_of_field}, its second-order Lie bracket averaged system is given as 
			\begin{subequations}
				\label{eq:LBS_newton_seeker}
				\begin{align}
					\dot{\bar{z}} &=  (J_0 + \Lambda \bar{d} ) \bar{z}   \\
					\dot{\bar{d}} &= \omega_d \bar{d} (1-H\bar{d}) \\
					\dot{\bar{\nu}} &= h(F(\bar{z}) - \bar{\nu}).
				\end{align}
			\end{subequations}
			where $J_0$,  $F(\bar{z})$ and  $\Lambda$ are defined as in \eqref{eq:J_0_def}, \eqref{eq:field_in_z} and \eqref{eq:Lambda_def}, respectively.
		\end{proposition}
		The proof of Proposition \ref{prop:average_newton_seeker} is provided in the Appendix. By linearizing \eqref{eq:LBS_gradient_seeker} around $\bar{z}^* = 0,\bar{d}^* = H^{-1}, \bar{\nu} = F^*$, we obtain 
		\begin{align}
			\begin{bmatrix}
				\dot{\bar{z}} \\
				\dot{\bar{d}} \\
				\dot{\bar{\nu}}
			\end{bmatrix} = \begin{bmatrix}
				0 &\omega_0 &0 &0 \\
				-\omega_0 &-\frac{\alpha}{2} &0 &0 \\
				0 &0 &-\omega_d &0 \\
				0 &0 &0 &-h
			\end{bmatrix} 	\begin{bmatrix}
				\bar{z} \\
				\bar{d} \\
				\bar{\nu}
			\end{bmatrix} = J_{\text{n}} \begin{bmatrix}
				\bar{z} \\
				\bar{d} \\
				\bar{\nu}
			\end{bmatrix} \label{eq:newton_seeker_lin}
		\end{align} where the eigenvalues of \eqref{eq:newton_seeker_lin} are, 
		\begin{align}
			\lambda_{1,2}(J_{\text{n}}) &= -\frac{\alpha }{2} \pm \frac{1}{2} \sqrt{ \left(\frac{\alpha }{2}\right)^2 - (2\omega_0)^2} \label{eq:ev_zbar_newton} \\
			\lambda_{3}(J_{\text{n}}) &= -\omega_d\\
			\lambda_{4}(J_{\text{n}}) &= -h.
		\end{align}
		From equation \eqref{eq:ev_zbar_newton}, we observe that $\bar{z}$ converges to zero, but in contrast to \eqref{eq:ev_zbar_grad}, the convergence is independent of the unknown Hessian matrix $H$.  The stability properties of the newly proposed Newton source seeking scheme \eqref{eq:newton_seeker} are analyzed in the following section.
		\subsection{Analysis}
		
		%{\color{red}
			\begin{theorem}[Newton source seeking]
				\label{thm:newton_ss_theorem}
				Consider \eqref{eq:newton_seeker} and let $c = 
				\omega^{1-p}$ and $\tilde{\alpha} = \alpha \omega^p$ where $p \in (0.5,1)$. Then, under  Assumption \ref{ass:shape_of_field}, the  point $x = x^*, d = H^{-1}, 
				\nu = F^*$ of \eqref{eq:newton_seeker} is sGPUAS as in Def. 
				\ref{def:sGPUAS} on the set $\{ d>0\}$, namely, for initial conditions $(x_0,d_0,\nu_0)$ where $|x_0|, |\nu_0|, d_0$ are arbitrarily large and $d_0>0$ is arbitrarily small.
			\end{theorem}

			\begin{proof}[Proof of Theorem \ref{thm:newton_ss_theorem}]
				The proof is divided into 5 steps to increase readability.
				
				{\noindent \textbf{Step 1.} \underline{\textit{Variable transformations}:}}
				
				In order to use the results from Section \ref{sec:preliminaries} for the convergence analysis of \eqref{eq:newton_seeker}, we need to bring the system in a suitable form. To this end, we first apply \eqref{eq:z_transformation} to \eqref{eq:newton_seeker} which was already done above and obtain the Newton source seeking scheme in $z$ as \eqref{eq:newton_seeker_in_z}.
				Next, we eliminate the unstable point $d = 0$ in the Riccati filter \eqref{eq:ricatti_filter_z} by applying the transformation 
				\begin{equation}
					\label{eq-d-to-dtilde}
					\tilde{d} = \ln(d), \qquad d>0 
				\end{equation}
				to \eqref{eq:ricatti_filter_z}. Thus, we obtain 
				{\setlength{\mathindent}{0pt}
					\begin{subequations}
						\label{eq:newton_seeker_in_z_and_d}
						\begin{align}
							\dot{z} &= J_0 z + \Big[ c e^{\tilde{d}} (F(z) - \nu) \sin(\omega t)  + \tilde{\alpha}  \cos(\omega t)\Big]  \begin{bmatrix}
								0 \\
								1
							\end{bmatrix} \\
							\dot{\tilde{d}} &= \omega_d \left(1 - e^{\tilde{d}} \frac{8 \omega^2}{\tilde{\alpha}^2}(F(z) - \nu) \cos(2\omega t) \right) \\
							\dot{\nu} &= h(F(z) - \nu).
						\end{align}
					\end{subequations}
				}

				{\noindent \textbf{Step 2.} \underline{\textit{Second-order Lie bracket averaging}:}}
				
				In this step, we calculate the second-order Lie bracket averaged system for \eqref{eq:newton_seeker_in_z_and_d}. Since we follow a similar procedure as in the derivation of \eqref{eq:LBS_newton_seeker} (see Appendix), we do not repeat the calculations here. By choosing $c = \omega^{1-p}$ and $\tilde{\alpha} = \alpha \omega^{p}$, where $p \in (0.5,1)$, the second-order Lie bracket averaged system is obtained as
				\begin{subequations}
					\label{eq:LBS_newton_seeker_exp}
					\begin{align}
						\dot{\bar{z}} &= (J_0 + \Lambda e^{\bar{\tilde{d}}}) \bar{z} \\
						\dot{\bar{\tilde{d}}} &= \omega_d (1 - H e^{\bar{\tilde{d}}}) \\
						\dot{\bar{\nu}} &= h(F(\bar{z}) - \bar{\nu}).
					\end{align}
				\end{subequations}
				where the above result holds under Assumption \ref{ass:shape_of_field}. 
				To analyze \eqref{eq:LBS_newton_seeker_exp} at the origin, we shift the coordinates as 
				\begin{align}
					\hat{d} &= \bar{\tilde{d}} - \ln(H^{-1})\\
					r &= \bar{\nu} - F(z)
				\end{align}
				and obtain 
				\begin{subequations}
					\label{eq:LBS_newton_seeker_exp_origin}
					\begin{align}
						\dot{r} &= -h r + H \bar{z}^\top(J_0 + \tilde{\Lambda} e^{\hat{d}}) \bar{z} \label{eq:sub_r}\\
						\dot{\bar{z}} &= (J_0 + \tilde{\Lambda} ) \bar{z} + (e^{\hat{d}} - 1) \tilde{\Lambda}  \bar{z} \label{eq:sub_z} \\
						\dot{\hat{d}} &= -\omega_d (e^{\hat{d}} - 1) \label{eq:sub_d} 
					\end{align}
				\end{subequations}
				where $\tilde{\Lambda} = H^{-1} \Lambda$. Observe that, in order to utilize Lemma \ref{lemma:sGPUAS} for the convergence analysis of \eqref{eq:newton_seeker_in_z_and_d}, we need to prove that \eqref{eq:LBS_newton_seeker_exp_origin} is GUAS. This would not be possible without transformation \eqref{eq-d-to-dtilde}.

				{\noindent \textbf{Step 3.} \underline{\textit{Stability analysis of the Lie bracket system}:}}
				
				The intentional choice of writing system \eqref{eq:LBS_newton_seeker_exp_origin} in its current form is to facilitate the recognition of its cascade system structure. This, in turn, allows us to utilize \cite[Lemma 4.7]{khalil2002nonlinear} for the analysis of its asymptotic stability.
				For this, we prove that the origin of subsystem \eqref{eq:sub_z} and \eqref{eq:sub_d} is GUAS. Then, we show that \eqref{eq:sub_r} is input-to-state stable (ISS)\footnote{ISS definition can be found in \cite[Def. 4.7]{khalil2002nonlinear}.} with respect to $\bar{z}$ and $\hat{d}$. Thus, we can conclude that the origin of \eqref{eq:LBS_newton_seeker_exp_origin} is GUAS.
				Inspired by \cite{ali_Lyapunov} and \cite{krstic1998useful}, consider a Lyapunov function for subsystem \eqref{eq:sub_z} and \eqref{eq:sub_d} as 
				\begin{align}
					V = \ln(1 + \bar{z}^\top P \bar{z})  + \frac{b}{\omega_d} (e^{\hat{d}}-\hat{d} - 1) \label{eq:Lyapunov_func}
				\end{align}
				where $P = P^T > 0$ is chosen as 
				\begin{equation}
					P = \begin{bmatrix}
						\frac{\alpha}{4 \omega_0^2} + \frac{2}{\alpha} &\frac{1}{2\omega_0} \\
						\frac{1}{2\omega_0} &\frac{2}{\alpha}
					\end{bmatrix},
				\end{equation}such that $P(J_0 + \tilde{\Lambda} ) + (J_0+ \tilde{\Lambda})^\top P = -I$ holds and the parameter $b$ is 
				\begin{equation}
					b = \frac{\norm{G}^2}{2 \lambda_{\text{min}} (P)}  > 0 \label{eq:b_const}
				\end{equation}
				with 
				\begin{equation}
					G = P \tilde{\Lambda} + \tilde{\Lambda}^\top P = \begin{bmatrix}
						0 &\frac{-\alpha}{4 \omega_0} \\
						\frac{-\alpha}{4 \omega_0} &-2
					\end{bmatrix} \label{eq:matrix_G}
				\end{equation}
				and 
				\begin{align}
					\norm{G}^2 &= 2 + \frac{\alpha^2}{16 \omega_0^2}  + \frac{1}{2 \omega_0} \sqrt{\alpha^2 + 16 \omega_0^2} \\
					\lambda_{\text{min}} (P) &= \frac{2}{\alpha} + \frac{\alpha}{8 \omega_0^2}  - \frac{1}{2 \omega_0^2} \sqrt{ \frac{\alpha^2}{16} +  \omega_0^2} >0 .
				\end{align}
				Then, the derivative of \eqref{eq:Lyapunov_func} along the trajectories of the subsystem \eqref{eq:sub_z} and \eqref{eq:sub_d} is 
				{\setlength{\mathindent}{0pt}
					\begin{equation}
						\begin{aligned}[b]
							\dot{V} =& \frac{\bar{z}^\top (P(J_0 + \tilde{\Lambda} ) + (J_0+ \tilde{\Lambda})^\top P) \bar{z}  }{1 + \bar{z}^\top P \bar{z}} \\  &+  \frac{(e^{\hat{d}} - 1)\bar{z}^\top (P\tilde{\Lambda}  +  \tilde{\Lambda}^\top P) \bar{z} }{1 + \bar{z}^\top P \bar{z} }- b(e^{\hat{d}} - 1)^2
						\end{aligned}
					\end{equation}
				}
				
				With the choice of $P$ and $G$ as above, $\dot{V}$ is written as 
				{\setlength{\mathindent}{0pt}
					\begin{equation}
						\label{eq:Vdot_temp}
						\begin{aligned}[b]
							&\dot{V} = \frac{-\abs{\bar{z}}^2 + (e^{\hat{d}} - 1)\bar{z}^\top G \bar{z} }{1 + \bar{z}^\top P \bar{z}}  - b(e^{\hat{d}} - 1)^2 \\
							&= \frac{-\abs{\bar{z}}^2 + (e^{\hat{d}} - 1)\bar{z}^\top G \bar{z} - b(e^{\hat{d}} - 1)^2 (1 + \bar{z}^\top P \bar{z})}{1 + \bar{z}^\top P \bar{z}}  
						\end{aligned}
					\end{equation}
				}Using the fact that $\bar{z}^\top P \bar{z} \le \lambda_{\text{min}}(P) \abs{\bar{z}}^2 $ and similarly $\bar{z}^\top G \bar{z} \le \norm{G} \abs{\bar{z}}^2 $, where $\lambda_{\text{min}}(P)$ is the smallest eigenvalue of the positive definite matrix $P$ and $\norm{G}$ is the matrix norm of \eqref{eq:matrix_G}, we can derive an upper bound for \eqref{eq:Vdot_temp} as 
				\begin{equation}
					\begin{aligned}[b]
						\dot{V} \le &\frac{-\abs{\bar{z}}^2 - b(e^{\hat{d}} - 1)^2    }{1 + \bar{z}^\top P \bar{z}} \\ &+ \frac{(e^{\hat{d}} - 1)\norm{G} \abs{\bar{z}}^2  - b(e^{\hat{d}} - 1)^2 \lambda_{\text{min}}(P) \abs{\bar{z}}^2}{1 + \bar{z}^\top P \bar{z}}
					\end{aligned}
				\end{equation}
				Taking the inequality
				\begin{equation}
					(e^{\hat{d}} - 1)\norm{G} \abs{\bar{z}}^2 \leq \frac{1}{2} \abs{\bar z}^2 + \frac{1}{2}(e^{\hat{d}} - 1)^2 \norm{G}^2 \abs{\bar{z}}^2
				\end{equation}
				into account and substituting it above, we obtain 
				\begin{equation}
					\begin{aligned}[b]
						\dot{V} \le &\frac{-\frac{1}{2}\abs{\bar{z}}^2 - b(e^{\hat{d}} - 1)^2  }{1 + \bar{z}^\top P \bar{z}} \\ &- \frac{(e^{\hat{d}} - 1)^2 \abs{\bar{z}}^2\left(b \lambda_{\text{min}}(P) - \frac{\norm{G}^2}{2} \right)  }{1 + \bar{z}^\top P \bar{z}}.
					\end{aligned}
				\end{equation}
				With the choice of $b$ as in \eqref{eq:b_const}, the upper bound for $\dot{V}$ reads as 
				\begin{equation}
					\dot{V} \le \frac{-\frac{1}{2}\abs{\bar{z}}^2 - b(e^{\hat{d}} - 1)^2  }{1 + \bar{z}^\top P \bar{z}}. \label{eq:Lyapunov_inequality}
				\end{equation}
				Since $\dot{V}$ is bounded from above by a negative definite function as in \eqref{eq:Lyapunov_inequality}, then according to \cite[Thm. 4.9]{khalil2002nonlinear}, we can conclude that the origin of the subsystem \eqref{eq:sub_z} and \eqref{eq:sub_d} is GUAS. 
				
				To complete the analysis, we need to demonstrate that \eqref{eq:sub_r} is ISS with respect to $\bar{z}$ and $\hat{d}$. To simplify the notation, we define the vector $g = [\bar{z}^\top, \hat{d}]^\top$ and our goal is to prove that \eqref{eq:sub_r} is ISS with respect to $g$. Hence, we can bound $\dot{r}$ as 
				\begin{equation}
					\begin{aligned}[b]
						\dot{r} &\le -h \abs{r} + H \abs{z}^2 \left(\norm{J_0} + e^{\hat{d}} \|\tilde{\Lambda}\| \right) \\
						&\le -h \abs{r} + H (\abs{z} + |\hat{d}|)^2 \left(\norm{J_0} + e^{(\abs{z} + |\hat{d}|)} \|\tilde{\Lambda}\| \right) \\
						&\le -h \abs{r} + H (2 \abs{g})^2 \left(\norm{J_0} + e^{2\abs{g}} \|\tilde{\Lambda}\| \right) \\
						&\le - \rho_1(\abs{r}) + \rho_2(\abs{g})	
					\end{aligned}
				\end{equation}
				Given that both $\rho_1$ and $\rho_2$ are $\mathcal{K}_{\infty}$ functions, we can refer to \cite[Thm. C.3]{krstic1995nonlinear} to conclude that the system \eqref{eq:sub_r} is ISS with respect to $g$. Considering this and the fact that subsystem \eqref{eq:sub_z} and \eqref{eq:sub_d} is GUAS, we can invoke \cite[Lemma 4.7]{khalil2002nonlinear} and conclude that the origin of \eqref{eq:LBS_newton_seeker_exp_origin} is GUAS.

				{\noindent \textbf{Step 4.} \underline{\textit{Second-order Lie bracket stability theorem}:}}
				
				In Step 3, we showed the asymptotic stability properties of the
				Lie bracket system \eqref{eq:LBS_newton_seeker_exp_origin}.  Hence, we can apply Lemma \ref{lemma:sGPUAS} from Section \ref{sec:preliminaries} and conclude
				that the equilibrium point $z^* = 0, \tilde{d} = \ln(H^{-1}), \nu = F^*$ of \eqref{eq:newton_seeker_in_z_and_d} is sGPUAS.

				%{\color{red} 
					\noindent \textbf{Step 5.} \underline{\textit{Back to the original coordinates}:}
					
					Invoking \eqref{eq-inverse-of-x-to-z} and \eqref{eq-d-to-dtilde}, we recall that $x=x^* +Y(t) z$ and $e^{\tilde{d}}$, respectively. From the fact that $z^* = 0, \tilde{d} = \ln(H^{-1}), \nu = F^*$ is sGPUAS on $\R^4$ for the $(z,\tilde d, \nu)$-system, we obtain that $x = x^*, d = H^{-1}, 
					\nu = F^*$ is sGPUAS on the set $\{ d>0\}$ for the $(x,d,\nu)$-system.
				\end{proof}
				
				\begin{remark}\em
					Theorem \ref{thm:newton_ss_theorem} solely provides the practical stability properties of \eqref{eq:newton_seeker}. However, the Hessian-invariant convergence rate is observed through the eigenvalues  \eqref{eq:ev_zbar_newton}.%{eq:LBS_newton_seeker}.
				\end{remark}

				%%%%%%%%%%%%%%%%%%%%%%%%%%%%%%%%%%%%%%%%%%%%%%%%%%%%%%%%%
				%%%%%%%%%%%%%%% New Section %%%%%%%%%%%%%%%%%%%%%%%%%%%%%
				%%%%%%%%%%%%%%%%%%%%%%%%%%%%%%%%%%%%%%%%%%%%%%%%%%%%%%%%%

				\section{Simulation Results}
				\label{sec:simulation}
				In this section, we evaluate the Newton source seeking scheme \eqref{eq:newton_seeker} on practical examples. We also refer to \eqref{eq:newton_seeker} as Newton seeker. To highlight the advantages of our proposed design, we compare it with the gradient source seeking scheme \eqref{eq:gradient_seeker}, analogously referred to as Gradient seeker, for the same values of the coinciding parameters. For the simulations, these are chosen as $\omega = 15$, $\omega_0 = 1$, $\alpha = 2$, $\omega_d = 0.3$ and $p = 0.61$.
				The field $F(x)$ is of the form \eqref{eq:shapeofField} with a peak of $F^* = 5$ and the source is located at $x^* = [1,-1]^\top$. Furthermore, the scalar Hessian parameter is $H = \frac{1}{100}$. The initial conditions are chosen as $x_0 = [4,-4]^\top$, $\nu_0 = 0$ and $d_0 = 1$.
				
				The  results are presented in Figures \ref{fig:hessian} and  \ref{fig:states_and_F}. Figure \ref{fig:hessian} clearly demonstrates that the Riccati filter successfully estimates the inverse of the Hessian, $H^{-1} = 100$. As a result, the Newton source seeking scheme \eqref{eq:newton_seeker} converges much faster towards the source $x^* = [1,-1]$, as observed in Figure \ref{fig:states_and_F}. In contrast, the gradient source seeking scheme \eqref{eq:gradient_seeker} exhibits a considerably slower convergence rate towards the source.
				
				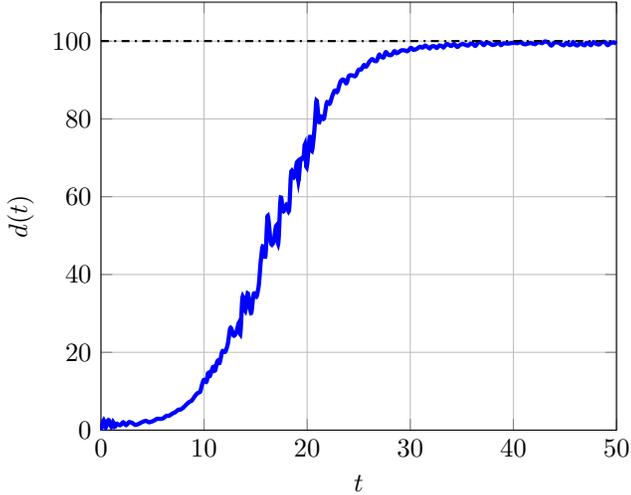
\begin{figure}
					\def\file{plots/newton_ss_data.txt}
\hspace*{-0.5cm}
\vspace*{-0.5cm}
\begin{tikzpicture}
	\begin{axis}[
		grid=both,
	xmin=0, xmax=50,
	ymin=0, ymax=110,
	ylabel=$d(t)$, xlabel=$t$]
	
	\addplot[ultra thick,blue,smooth] table[x=t_grad,y=d]{\file};

	\addplot[thick,black,dashdotted, domain = 0:50,samples = 2] {100};
\end{axis};
\end{tikzpicture}
					\caption{The Riccati state $d(t)$ converging to the inverse of the Hessian $H^{-1} = 100$.}
					\label{fig:hessian}
				\end{figure}

				\begin{figure}
					%	\centering
					
					\begin{subfigure}[b]{0.55\textwidth}
						\includegraphics[scale = 0.3]{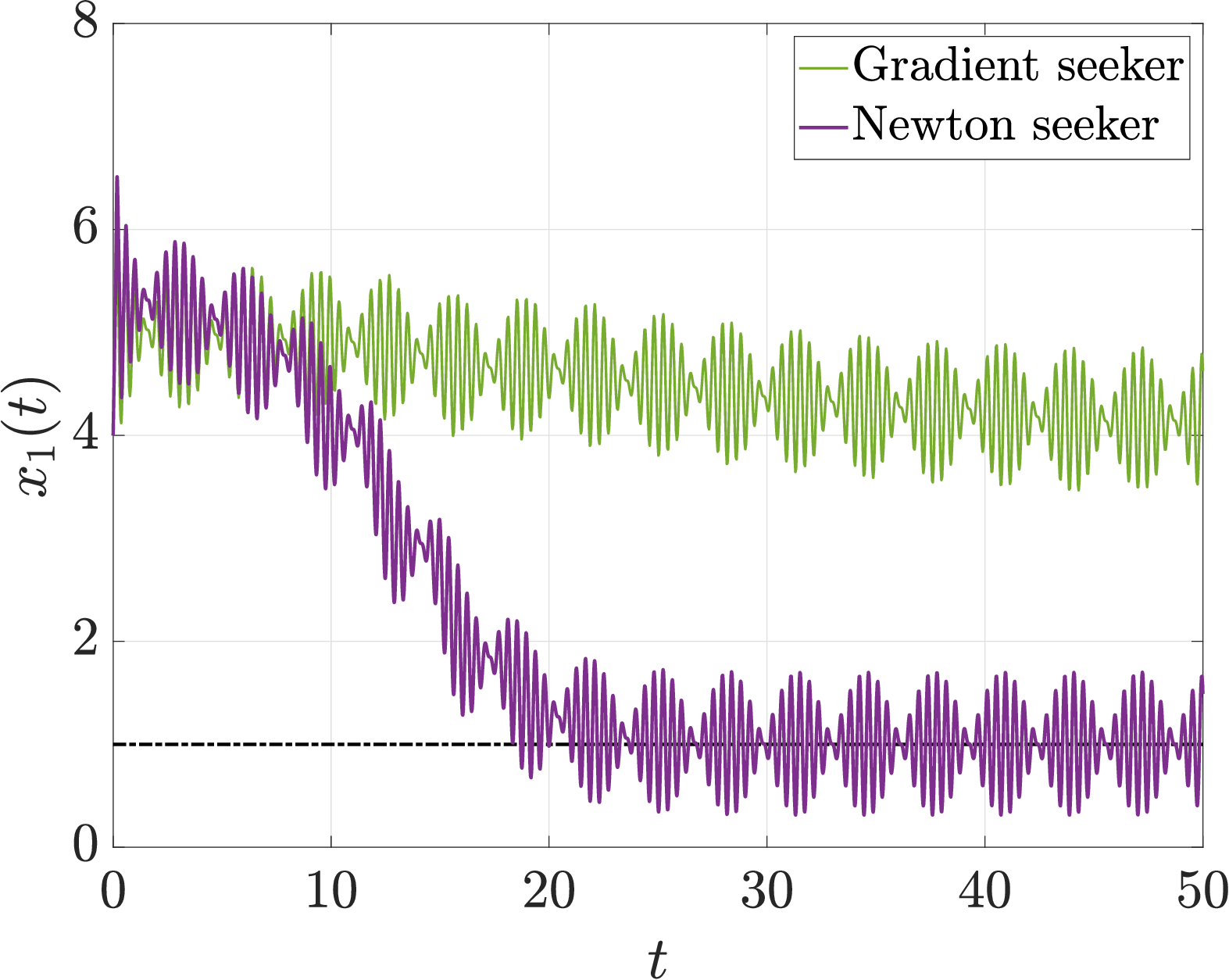} 
						\caption{}
					\end{subfigure}
					
					\begin{subfigure}[b]{0.55\textwidth}
						\includegraphics[scale = 0.3]{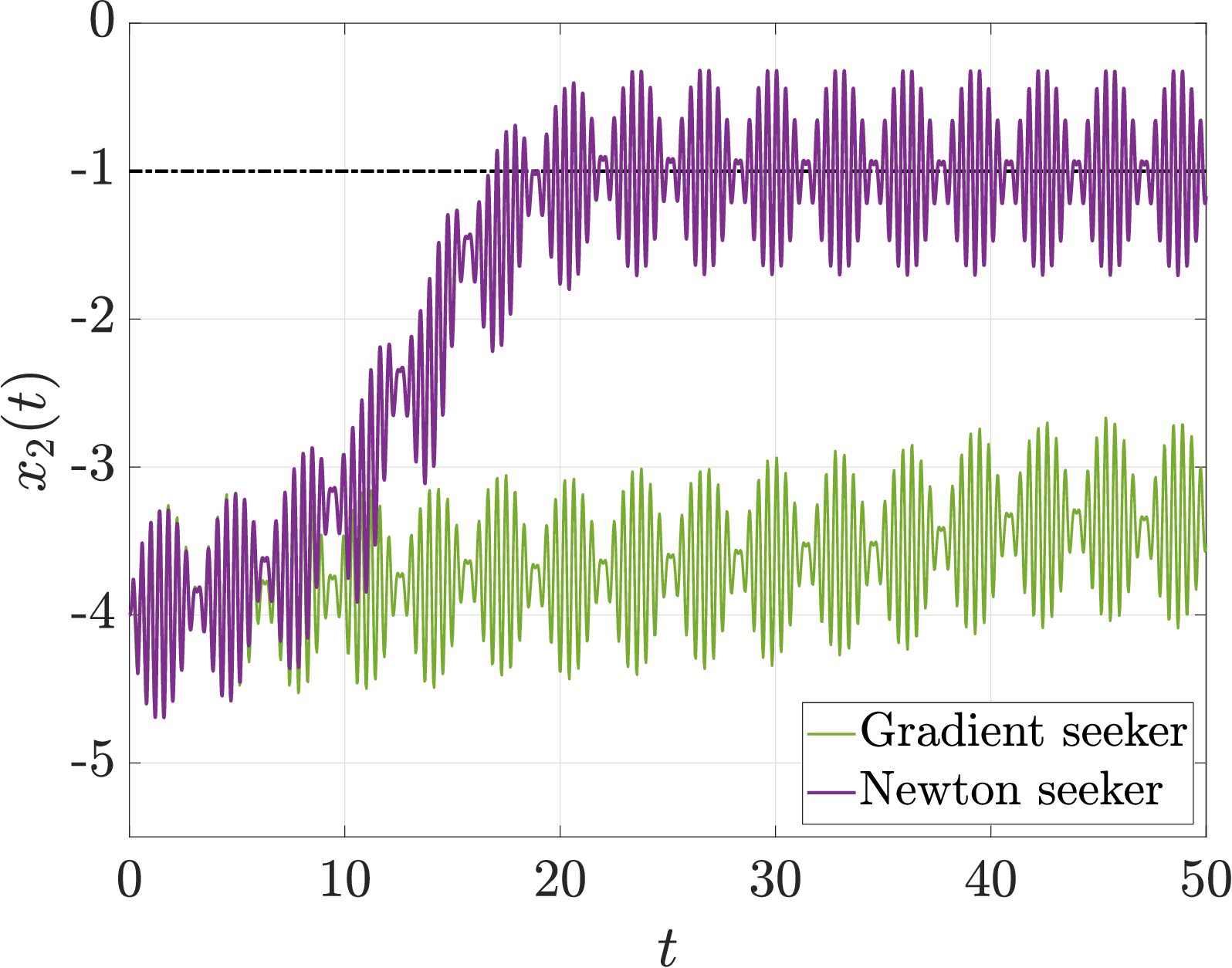}
						\caption{}
						\label{fig:x}
					\end{subfigure}

					\caption{The trajectories of both source seeking schemes, the gradient seeker \eqref{eq:gradient_seeker}  and Newton seeker \eqref{eq:newton_seeker} where: (a) $x_1$-trajectory of \eqref{eq:newton_seeker} converges to $x_1^* = 1$ and (b) $x_2$-trajectory of \eqref{eq:newton_seeker} converges to $x_2^* = -1$.}
					\label{fig:states_and_F}
				\end{figure}

				%\end{enumerate}
				%	\begin{figure}[t]
					%		\hspace*{-0.5cm}
					%		\includegraphics[trim = {2.5cm 0cm 0cm 1cm},clip,scale = 0.3]{images/nu.eps}
					%		\caption{The 'clipping' of $\nu(t-t_0)$ to $0.3$ at time $t = 0.7T$ shown left and the resulting time-varying gain \eqref{eq:kt} shown right. }
					%		\label{fig:nu}
					%	\end{figure} 

				%{\color{red}
					\section{Conclusion}
					
					While it may seem the most desirable to develop a Newton-based version of a source seeking algorithm with steering-based source seeking, such as, for example, \cite{ES_with_bounded_update_rates}, a naive Newton-based modification of this algorithm results in a second-order Lie bracket system %\eqref{eq:second_order_LBS_general} 
					that does not converge.
					Hence, we develop the first Newton-based source seeking algorithm based on the historically first nonholonomic source seeking scheme \cite{zhang2007source}, which employs forward velocity tuning. 
					Our convergence rate is independent from the Hessian and can be assigned by a proper choice of the gains. This is done by employing a Riccati filter  that, on the average, estimates the inverse of the Hessian of the distance-dependent field, and modifies the gradient algorithm \cite{zhang2007source}. 
					We prove semi-global practical stability properties and confirm our theoretical results with simulations.
					
					The need and value of second-order averaging in source seeking was recently revealed in the award-winning conference paper \cite{10156319} and is affirmed by the present paper. 
					%%%%%%%%%%%%%%%%%%%%%%%%%%%%%%%%%%%%%%%%%%%%%%%%%%%%%%%%%%%%%%%%%%%%%%%%%%%%%%%%

					%%%%%%%%%%%%%%%%%%%%%%%%%%%%%%%%%%%%%%%%%%%%%%%%%%%%%%%%%%%%%%%%%%%%%%%%%%%%%%%%

					%%%%%%%%%%%%%%%%%%%%%%%%%%%%%%%%%%%%%%%%%%%%%%%%%%%%%%%%%%%%%%%%%%%%%%%%%%%%%%%%

					\section*{Appendix}
					
					\label{appendix}
					
					\begin{proof}[Proof of Proposition \ref{prop:average_gradient_seeker}]
						For the derivation of \eqref{eq:LBS_gradient_seeker}, we first bring \eqref{eq:gradient_seeker_in_z} in the input-affine form \eqref{eq:general_nominal_system}. Let $c = \omega^{1-p}$ and $ \tilde{\alpha}= \alpha \omega^p$ where $p \in (0.5,1)$, then we can write \eqref{eq:gradient_seeker_in_z} as
						\begin{equation}
							\begin{aligned}[b]
								\begin{bmatrix}
									\dot{z} \\
									\dot{\nu}
								\end{bmatrix} = &\underbrace{\begin{bmatrix} J_0 z \\ h(F(z) - \nu) \end{bmatrix}}_{f_0} + \underbrace{\begin{bmatrix}
										0 \\ F(z)-\nu \\0
								\end{bmatrix}}_{f_1} \omega^{1-p} u_1(\omega t) \\ &+ \underbrace{\begin{bmatrix}
										0 \\ \alpha \\0
								\end{bmatrix}}_{f_2} \omega^p u_2(\omega t) 
							\end{aligned}
						\end{equation}
						with $u_1(\omega t) = \sin(\omega t)$, $u_2(\omega t) = \cos(\omega t)$, $p_1 = 1-p$ and $p_2 = p$. The derivation of \eqref{eq:LBS_gradient_seeker} follows from a direct calculation. With $T = \frac{2 \pi}{\omega}$ and \eqref{eq:field_in_z},  the first order terms are obtained as
						\begin{equation}
							\gamma_{12}(\omega) = \frac{\omega^2}{2 \pi} \int\limits_{0}^{\frac{2 \pi}{\omega}} \int\limits_{0}^{s} \cos(\omega s) \sin(\omega p) \text{d}p \text{d} s = -\frac{1}{2}
						\end{equation}
						and
						\begin{equation}
							[f_1,f_2] = \begin{bmatrix}
								0 \\
								\alpha H z_2 \\
								0
							\end{bmatrix}
						\end{equation}
						The second-order terms are obtained in similar fashion as
						{\setlength{\mathindent}{0pt}
							\begin{align}
								\begin{aligned}[b] \gamma_{121}(\omega) =
									\frac{\omega^{3-p}}{6 \pi} &\int\limits_{0}^{\frac{2\pi}{\omega}} \int\limits_{0}^{\tau}  \int\limits_{0}^{s} \sin ( \omega \tau)  ( \cos ( \omega s) \sin( \omega p)  \\ &- \sin ( \omega s) \cos( \omega p) ) \text{d} p \text{d} s \text{d} \tau   = \frac{1}{2\omega^p}
								\end{aligned}
							\end{align}
						}
						{\setlength{\mathindent}{0pt}
							\begin{align}
								\begin{aligned}[b] \gamma_{122}(\omega) =
									\frac{\omega^{2+p}}{6 \pi} &\int\limits_{0}^{\frac{2\pi}{\omega}} \int\limits_{0}^{\tau}  \int\limits_{0}^{s} \cos ( \omega \tau)  ( \cos( \omega s) \sin( \omega p)  \\ &- \sin ( \omega s) \cos( \omega p) ) \text{d} p \text{d} s \text{d} \tau   = 0
								\end{aligned}
							\end{align}
						}and
						\begin{equation}
							[[f_1,f_2],f_1] =  \begin{bmatrix}
								0 \\
								\alpha H (\nu - F(z)) - \alpha H^2 z_2^2 \\
								0
							\end{bmatrix}.
						\end{equation}
						Since $\gamma_{122}(\omega) = 0$, the value of $[[f_1,f_2],f_2]$ does not contribute in the expression for \eqref{eq:LBS_gradient_seeker} and thus we omit its computation.
						Finally, multiplying and adding the first and second-order terms from above as in \eqref{eq:second_order_LBS_general} with $l =  2$ and taking the limit as $\omega \rightarrow \infty$, we obtain \eqref{eq:LBS_gradient_seeker}.
					\end{proof}

					\begin{proof}[Proof of Proposition \ref{prop:average_newton_seeker}]
						\label{proof:average_newton_seeker}
						Following a similar approach to the proof of Proposition \ref{prop:average_gradient_seeker}, we can express \eqref{eq:newton_seeker_in_z} in the form of \eqref{eq:general_nominal_system}. By choosing $c = \omega^{1-p}$ and $\tilde{\alpha} = \alpha \omega^p$, where $p \in (0.5,1)$, we rewrite \eqref{eq:newton_seeker_in_z} as
						{\setlength{\mathindent}{0pt}
							\begin{equation}
								\begin{aligned}[b]
									&\begin{bmatrix}
										\dot{z} \\
										\dot{d} \\
										\dot{\nu}
									\end{bmatrix} = \underbrace{\begin{bmatrix} J_0 z \\ \omega_d  d \\ h(F(z) - \nu) \end{bmatrix}}_{f_0} + \underbrace{\begin{bmatrix}
											0 \\ d( F(z)-\nu) \\0 \\0
									\end{bmatrix}}_{f_1} \omega^{1-p} u_1(\omega t) \\ &+ \underbrace{\begin{bmatrix}
											0 \\ \alpha \\0 \\0
									\end{bmatrix}}_{f_2} \omega^p u_2(\omega t) + \underbrace{\begin{bmatrix}
											0 \\ 0\\ d^2 \frac{8\omega_d}{\alpha^2}(\nu - F(z))  \\0
									\end{bmatrix}}_{f_3} \omega^{2-2p} u_3(\omega t)
								\end{aligned}
							\end{equation}
						}with $u_1(\omega t) = \sin(\omega t)$, $u_2(\omega t) = \cos(\omega t)$, $u_3(\omega t) = \cos(2 \omega t)$, , $p_1 = 1-p$, $p_2 = p$ and $p_3 = 2-2p$.
						With $T = \frac{2 \pi}{\omega}$ and \eqref{eq:field_in_z}, the first order terms are obtained as
						\begin{align}
							\gamma_{12}(\omega) &= \frac{\omega^2}{2 \pi} \int\limits_{0}^{\frac{2 \pi}{\omega}} \int\limits_{0}^{s} \cos(\omega s) \sin(\omega p) \text{d}p \text{d} s = -\frac{1}{2} \\
							\gamma_{13}(\omega) &= \frac{\omega^{3-3p}}{2 \pi} \int\limits_{0}^{\frac{2 \pi}{\omega}} \int\limits_{0}^{s} \cos(2 \omega s) \sin(\omega p) \text{d}p \text{d} s = 0 \\
							\gamma_{23}(\omega) &= \frac{\omega^{3-p}}{2 \pi} \int\limits_{0}^{\frac{2 \pi}{\omega}} \int\limits_{0}^{s} \cos(2 \omega s) \cos(\omega p) \text{d}p \text{d} s = 0
						\end{align}
						and 
						\begin{align}
							[f_1,f_2] &= \begin{bmatrix}
								0 
								&d \alpha H  z_2 
								&0 
								&0
							\end{bmatrix}^\top\\
							[f_1,f_3] &= \begin{bmatrix}
								0 \\ - d^2 \frac{8 \omega_d ((F(z) - F^*)^2 - \nu)}{ \alpha^2}\\ d^3 \frac{8 H \omega_d z_2 (F(z) - \nu) }{\alpha^2} \\0
							\end{bmatrix} \\
							[f_2,f_3] &= \begin{bmatrix}
								0 
								&0 
								&d^2 H \frac{8\omega_d}{\alpha}  z_2 
								&0
							\end{bmatrix}^\top   
						\end{align}
						Analogously, the second-order terms are computed as
						{\setlength{\mathindent}{0pt}
							\begin{align}
								\begin{aligned}[b] \gamma_{121}(\omega) =
									\frac{\omega^{3-p}}{6 \pi} &\int\limits_{0}^{\frac{2\pi}{\omega}} \int\limits_{0}^{\tau}  \int\limits_{0}^{s} \sin ( \omega \tau)  ( \cos ( \omega s) \sin( \omega p)  \\ &- \sin ( \omega s) \cos( \omega p) ) \text{d} p \text{d} s \text{d} \tau   = \frac{1}{2\omega^p}
								\end{aligned}
							\end{align}
							\begin{align}
								\begin{aligned}[b] \gamma_{122}(\omega) =
									\frac{\omega^{2+p}}{6 \pi} &\int\limits_{0}^{\frac{2\pi}{\omega}} \int\limits_{0}^{\tau}  \int\limits_{0}^{s} \cos ( \omega \tau)  ( \cos ( \omega s) \sin( \omega p)  \\ &- \sin ( \omega s) \cos( \omega p) ) \text{d} p \text{d} s \text{d} \tau   = 0
								\end{aligned}
							\end{align}
							\begin{align}
								\begin{aligned}[b] \gamma_{123}(\omega) =
									\frac{\omega^{4-2p}}{6 \pi} &\int\limits_{0}^{\frac{2\pi}{\omega}} \int\limits_{0}^{\tau}  \int\limits_{0}^{s} \cos ( 2\omega \tau)  ( \cos ( \omega s) \sin( \omega p)  \\ &- \sin ( \omega s) \cos( \omega p) ) \text{d} p \text{d} s \text{d} \tau   = 0
								\end{aligned}
							\end{align}
							\begin{align}
								\begin{aligned}[b] \gamma_{131}(\omega) =
									&\frac{\omega^{5-4p}}{6 \pi} \int\limits_{0}^{\frac{2\pi}{\omega}} \int\limits_{0}^{\tau}  \int\limits_{0}^{s} \sin ( \omega \tau)  ( \cos ( 2 \omega s) \sin( \omega p)  \\ &- \sin ( \omega s) \cos( 2 \omega p) ) \text{d} p \text{d} s \text{d} \tau   = -\frac{1}{8\omega^{4p - 2}}
								\end{aligned}
							\end{align}
							\begin{align}
								\begin{aligned}[b] \gamma_{132}(\omega) =
									\frac{\omega^{4-2p}}{6 \pi} &\int\limits_{0}^{\frac{2\pi}{\omega}} \int\limits_{0}^{\tau}  \int\limits_{0}^{s} \cos ( \omega \tau)  ( \cos (2 \omega s) \sin( \omega p)  \\ &- \sin ( \omega s) \cos( 2 \omega p) ) \text{d} p \text{d} s \text{d} \tau   = 0
								\end{aligned}
							\end{align}
							\begin{align}
								\begin{aligned}[b] \gamma_{133}(\omega) =
									\frac{\omega^{6-5p}}{6 \pi} &\int\limits_{0}^{\frac{2\pi}{\omega}} \int\limits_{0}^{\tau}  \int\limits_{0}^{s} \cos ( 2\omega \tau)  ( \cos (2 \omega s) \sin( \omega p)  \\ &- \sin ( \omega s) \cos(2 \omega p) ) \text{d} p \text{d} s \text{d} \tau   = 0
								\end{aligned}
							\end{align}
							\begin{align}
								\begin{aligned}[b] \gamma_{231}(\omega) =
									&\frac{\omega^{4-2p}}{6 \pi} \int\limits_{0}^{\frac{2\pi}{\omega}} \int\limits_{0}^{\tau}  \int\limits_{0}^{s} \sin ( \omega \tau)  ( \cos ( 2 \omega s) \cos( \omega p)  \\ &- \cos( \omega s) \cos( 2 \omega p) ) \text{d} p \text{d} s \text{d} \tau   = 0
								\end{aligned}
							\end{align}
							\begin{align}
								\begin{aligned}[b] \gamma_{232}(\omega) =
									\frac{\omega^{3}}{6 \pi} &\int\limits_{0}^{\frac{2\pi}{\omega}} \int\limits_{0}^{\tau}  \int\limits_{0}^{s} \cos ( \omega \tau)  ( \cos (2 \omega s) \cos( \omega p)  \\ &- \cos ( \omega s) \cos( 2 \omega p) ) \text{d} p \text{d} s \text{d} \tau   = \frac{1}{8}
								\end{aligned}
							\end{align}
							\begin{align}
								\begin{aligned}[b] \gamma_{233}(\omega) =
									\frac{\omega^{5-3p}}{6 \pi} &\int\limits_{0}^{\frac{2\pi}{\omega}} \int\limits_{0}^{\tau}  \int\limits_{0}^{s} \cos ( 2\omega \tau)  ( \cos (2 \omega s) \cos( \omega p)  \\ &- \cos ( \omega s) \cos(2 \omega p) ) \text{d} p \text{d} s \text{d} \tau   = 0
								\end{aligned}
							\end{align}
						}As previously, we  skip the computation of the second-order terms $[[f_i,f_j], f_m]$ where their corresponding $\gamma_{ijm}(\omega) = 0$, as they do not contribute to the calculation of \eqref{eq:LBS_newton_seeker}. For instance, there is no need to compute $[[f_1,f_2], f_2]$ since $\gamma_{122}(\omega) = 0$. Therefore, the relevant second-order Lie brackets are 
						{\setlength{\mathindent}{0pt}
							\begin{align}
								[[f_1,f_2], f_1] &= \begin{bmatrix}
									0 \\
									H \alpha d^2 (\nu - F(z)) - \alpha (H d z_2)^2 \\
									0 \\
									0
								\end{bmatrix}\\
								[[f_1,f_3], f_1] &= \begin{bmatrix}
									0 \\ 
									- \frac{16 H d^3 \omega_d z_2}{\alpha^2}(F(z) - \nu)\\
									\frac{8 H d^4 (F(z) - \nu) }{\alpha^2}  \left(H z_2^2 - (F(z) - \nu)\right) \\ 
									\\0
								\end{bmatrix} \\
								[[f_2, f_3], f_2] &= \begin{bmatrix}
									0 \\
									0 \\
									-8 \omega_d H  d^2\\
									0
								\end{bmatrix}   
							\end{align}
						}By multiplying and adding the first and second-order terms  above, as described in \eqref{eq:second_order_LBS_general} with $l = 3$, and taking the limit as $\omega$ approaches infinity, we arrive at  \eqref{eq:LBS_newton_seeker}.

					\end{proof}

					\bibliography{references}

% Generated by IEEEtranS.bst, version: 1.14 (2015/08/26)
\begin{thebibliography}{10}
\providecommand{\url}[1]{#1}
\csname url@samestyle\endcsname
\providecommand{\newblock}{\relax}
\providecommand{\bibinfo}[2]{#2}
\providecommand{\BIBentrySTDinterwordspacing}{\spaceskip=0pt\relax}
\providecommand{\BIBentryALTinterwordstretchfactor}{4}
\providecommand{\BIBentryALTinterwordspacing}{\spaceskip=\fontdimen2\font plus
\BIBentryALTinterwordstretchfactor\fontdimen3\font minus
  \fontdimen4\font\relax}
\providecommand{\BIBforeignlanguage}[2]{{%
\expandafter\ifx\csname l@#1\endcsname\relax
\typeout{** WARNING: IEEEtranS.bst: No hyphenation pattern has been}%
\typeout{** loaded for the language `#1'. Using the pattern for}%
\typeout{** the default language instead.}%
\else
\language=\csname l@#1\endcsname
\fi
#2}}
\providecommand{\BIBdecl}{\relax}
\BIBdecl

\bibitem{abdelgalil2022sea}
M.~Abdelgalil, Y.~Aboelkassem, and H.~Taha, ``Sea urchin sperm exploit extremum
  seeking control to find the egg,'' \emph{Physical Review E}, vol. 106, no.~6,
  p. L062401, 2022.

\bibitem{10156319}
M.~Abdelgalil, A.~Eldesoukey, and H.~Taha, ``Singularly perturbed averaging
  with application to bio-inspired 3d source seeking,'' in \emph{2023 American
  Control Conference (ACC)}, 2023, pp. 885--890.

\bibitem{9782675}
M.~Abdelgalil and H.~Taha, ``Recursive averaging with application to
  bio-inspired 3-d source seeking,'' \emph{IEEE Control Systems Letters},
  vol.~6, pp. 2816--2821, 2022.

\bibitem{ali_Lyapunov}
A.~A. Ahmadi, M.~Krstic, and P.~A. Parrilo, ``A globally asymptotically stable
  polynomial vector field with no polynomial lyapunov function,'' in \emph{2011
  50th IEEE Conference on Decision and Control and European Control
  Conference}, 2011, pp. 7579--7580.

\bibitem{ES_book}
K.~B. Ariyur and M.~Krstic, \emph{Real-time optimization by extremum-seeking
  control}.\hskip 1em plus 0.5em minus 0.4em\relax John Wiley \& Sons, 2003.

\bibitem{cochran2009nonholonomic}
J.~Cochran and M.~Krstic, ``Nonholonomic source seeking with tuning of angular
  velocity,'' \emph{IEEE Transactions on Automatic Control}, vol.~54, no.~4,
  pp. 717--731, 2009.

\bibitem{LBA_first}
H.-B. D{\"u}rr, M.~S. Stankovi{\'c}, C.~Ebenbauer, and K.~H. Johansson, ``Lie
  bracket approximation of extremum seeking systems,'' \emph{Automatica},
  vol.~49, no.~6, pp. 1538--1552, 2013.

\bibitem{newton_ES}
A.~Ghaffari, M.~Krsti{\'c}, and D.~Ne{\v{s}}i{\'c}, ``Multivariable
  {N}ewton-based extremum seeking,'' \emph{Automatica}, vol.~48, no.~8, pp.
  1759--1767, 2012.

\bibitem{power_newton}
A.~Ghaffari, M.~Krsti{\'c}, and S.~Seshagiri, ``Power optimization for
  photovoltaic microconverters using multivariable {N}ewton-based extremum
  seeking,'' \emph{IEEE Transactions on Control Systems Technology}, vol.~22,
  no.~6, pp. 2141--2149, 2014.

\bibitem{6160495}
A.~Ghaffari, M.~Krstić, and D.~Nešić, ``Multivariable newton-based extremum
  seeking,'' in \emph{2011 50th IEEE Conference on Decision and Control and
  European Control Conference}, 2011, pp. 4436--4441.

\bibitem{khalil2002nonlinear}
H.~K. Khalil, \emph{{Nonlinear systems; 3rd ed.}}\hskip 1em plus 0.5em minus
  0.4em\relax Prentice hall, 2002.

\bibitem{krstic1998useful}
M.~Krstic, D.~Fontaine, P.~V. Kokotovic, and J.~D. Paduano, ``Useful
  nonlinearities and global stabilization of bifurcations in a model of jet
  engine surge and stall,'' \emph{IEEE Transactions on Automatic Control},
  vol.~43, no.~12, pp. 1739--1745, 1998.

\bibitem{krstic1995nonlinear}
M.~Krstic, P.~V. Kokotovic, and I.~Kanellakopoulos, \emph{Nonlinear and
  adaptive control design}.\hskip 1em plus 0.5em minus 0.4em\relax John Wiley
  \& Sons, Inc., 1995.

\bibitem{labar2019newton}
C.~Labar, E.~Garone, M.~Kinnaert, and C.~Ebenbauer, ``Newton-based extremum
  seeking: A second-order lie bracket approximation approach,''
  \emph{Automatica}, vol. 105, pp. 356--367, 2019.

\bibitem{lin2017stochastic}
J.~Lin, S.~Song, K.~You, and M.~Krstic, ``Stochastic source seeking with
  forward and angular velocity regulation,'' \emph{Automatica}, vol.~83, pp.
  378--386, 2017.

\bibitem{liu2010stochastic}
S.-J. Liu and M.~Krstic, ``Stochastic source seeking for nonholonomic
  unicycle,'' \emph{Automatica}, vol.~46, no.~9, pp. 1443--1453, 2010.

\bibitem{stochastic_newton}
------, ``Newton-based stochastic extremum seeking,'' \emph{Automatica},
  vol.~50, no.~3, pp. 952--961, 2014.

\bibitem{nesic2010newton}
D.~Nesi{\'c}, Y.~Tan, W.~H. Moase, and C.~Manzie, ``A unifying approach to
  extremum seeking: Adaptive schemes based on estimation of derivatives,'' in
  \emph{49th IEEE conference on decision and control (CDC)}.\hskip 1em plus
  0.5em minus 0.4em\relax IEEE, 2010, pp. 4625--4630.

\bibitem{PDE_newton}
T.~R. Oliveira, J.~Feiling, S.~Koga, and M.~Krsti{\'c}, ``Multivariable
  extremum seeking for {PDE} dynamic systems,'' \emph{IEEE Transactions on
  Automatic Control}, vol.~65, no.~11, pp. 4949--4956, 2020.

\bibitem{fixed_newton}
J.~I. Poveda and M.~Krstic, ``Fixed-time extremum seeking,'' \emph{arXiv
  preprint arXiv:1912.06999}, 2019.

\bibitem{damir_newton2}
D.~Ru{\v{s}}iti, T.~R. Oliveira, M.~Krsti{\'c}, and M.~Gerdts, ``Newton-based
  extremum seeking of higher-derivative maps with time-varying delays,''
  \emph{International Journal of Adaptive Control and Signal Processing},
  vol.~35, no.~7, pp. 1202--1216, 2021.

\bibitem{damir_newton}
D.~Ru{\v{s}}iti, T.~R. Oliveira, G.~Mills, and M.~Krsti{\'c}, ``Deterministic
  and stochastic {N}ewton-based extremum seeking for higher derivatives of
  unknown maps with delays,'' \emph{Europ. J. Contr.}, vol.~41, pp. 72--83,
  2018.

\bibitem{ES_with_bounded_update_rates}
A.~Scheinker and M.~Krsti{\'c}, ``Extremum seeking with bounded update rates,''
  \emph{Systems \& Control Letters}, vol.~63, pp. 25--31, 2014.

\bibitem{suttner2019extremum}
R.~Suttner, ``Extremum seeking control for an acceleration controlled
  unicycle,'' \emph{IFAC-PapersOnLine}, vol.~52, no.~16, pp. 676--681, 2019.

\bibitem{suttner2020acceleration}
R.~Suttner and M.~Krsti{\'c}, ``Acceleration-actuated source seeking without
  position and velocity sensing,'' \emph{IFAC-PapersOnLine}, vol.~53, no.~2,
  pp. 5348--5355, 2020.

\bibitem{suttner2023nonlocal}
R.~Suttner and M.~Krstic, ``Nonlocal nonholonomic source seeking despite local
  extrema,'' \emph{arXiv preprint arXiv:2303.16027}, 2023.

\bibitem{suttner_torque}
R.~Suttner and M.~Krstić, ``Source seeking with a torque-controlled
  unicycle,'' \emph{IEEE Control Systems Letters}, vol.~7, pp. 79--84, 2023.

\bibitem{todorovski2023practical}
V.~Todorovski and M.~Krstic, ``Practical prescribed-time seeking of a repulsive
  source by unicycle angular velocity tuning,'' \emph{Automatica}, vol. 154, p.
  111069, 2023.

\bibitem{ss_planar}
B.~Wang, S.~Nersesov, H.~Ashrafiuon, P.~Naseradinmousavi, and M.~Krstić,
  ``Source seeking for planar underactuated vehicles by surge force tuning,''
  in \emph{2022 IEEE 61st Conference on Decision and Control (CDC)}, 2022, pp.
  5335--5342.

\bibitem{zhang2007source}
C.~Zhang, D.~Arnold, N.~Ghods, A.~Siranosian, and M.~Krstic, ``Source seeking
  with non-holonomic unicycle without position measurement and with tuning of
  forward velocity,'' \emph{Systems \& control letters}, vol.~56, no.~3, pp.
  245--252, 2007.

\bibitem{first_ss}
C.~Zhang, A.~Siranosian, and M.~Krsti{\'c}, ``Extremum seeking for moderately
  unstable systems and for autonomous vehicle target tracking without position
  measurements,'' \emph{Automatica}, vol.~43, no.~10, pp. 1832--1839, 2007.

\end{thebibliography}

				\end{document}